\documentclass[conference]{IEEEtran}
\IEEEoverridecommandlockouts
\usepackage{cite}
\usepackage{amsmath,amssymb,amsfonts,mathtools}
\usepackage{amsthm}
\usepackage{algorithmic}
\usepackage{graphicx}
\usepackage{textcomp}
\usepackage{xcolor}
\mathtoolsset{showonlyrefs}
\newtheorem{theorem}{Theorem}
\newenvironment{definition}[1][Definition]{\begin{trivlist}
		\item[\hskip \labelsep {\bfseries #1}]}{\end{trivlist}}
\renewcommand{\baselinestretch}{0.917}
\def\BibTeX{{\rm B\kern-.05em{\sc i\kern-.025em b}\kern-.08em
    T\kern-.1667em\lower.7ex\hbox{E}\kern-.125emX}}
\begin{document}

\title{\huge On Reliability of Underwater Magnetic Induction Communications with Tri-Axis Coils
}

\author{\IEEEauthorblockN{ Hongzhi Guo$^1$, Zhi Sun$^2$, and Pu Wang$^3$}
\IEEEauthorblockA{$^1$Engineering Department,
	Norfolk State University,
	Email: hguo@nsu.edu\\ 
$^2$Electrical Engineering Department,
University at Buffalo, State University of New York,
Email: zhisun@buffalo.edu\\
$^3$Department of Computer Science,
University of North Carolina at Charlotte,
Email: pu.wang@uncc.edu
	}
}

\maketitle

\begin{abstract}
Underwater magnetic induction communications (UWMICs) provide a low-power and high-throughput solution for autonomous underwater vehicles (AUVs), which are envisioned to explore and monitor the underwater environment. UWMIC with tri-axis coils increases the reliability of the wireless channel by exploring the coil orientation diversity. However, the UWMIC channel is different from typical fading channels and the mutual inductance information (MII) is not always available. It is not clear the performance of the tri-axis coil MIMO without MII. Also, its performances with multiple users have not been investigated. In this paper, we analyze the reliability and multiplexing gain of UWMICs with tri-axis coils by using coil selection. We optimally select the transmit and receive coils to reduce the computation complexity and power consumption and explore the diversity for multiple users. We find that without using all the coils and MII, we can still achieve reliability. Also, the multiplexing gain of UWMIC without MII is 5dB smaller than typical terrestrial fading channels. The results of this paper provide a more power-efficient way to use UWMICs with tri-axis coils.  
\end{abstract}


\section{Introduction}

Wireless communication is a key technology for underwater exploration and monitoring using autonomous underwater vehicles (AUV). Acoustic and optical signals are popular solutions thanks to their long communication range and high data rate. Radio frequency (RF) signals are rarely used due to the high absorption of lossy water medium. Recent research shows that using magnetic induction with frequency from 10 KHz to 10 MHz is an efficient solution. Although this technology has been used for more than a hundred years for submarines, only recently it is applied to small robots in underwater wireless communications.     

The underwater magnetic induction communication (UWMIC) channel has been modeled in \cite{domingo2012magnetic,gulbahar2012communication,guo2017multiple}. In \cite{domingo2012magnetic}, the direct magnetic induction channel and waveguide channel with large coils are analytically derived using circuit analysis. In \cite{gulbahar2012communication}, a UWMIC network is designed with a large numbers of devices that are placed in cubes. The communication performances, such as bit error rate and reliability, are presented. Both \cite{domingo2012magnetic} and \cite{gulbahar2012communication} consider deep water by neglecting the water-air interface and employ advanced technologies to reduce the effects of high conductive water medium, such as super-conductive materials. In \cite{guo2017multiple}, the UWMIC is employed in shallow water and the water-air interface is treated. By leveraging the tri-axis coil, a multiple-input-multiple-output (MIMO) wireless link is analyzed; the results show that the tri-axis coil with ideal mutual inductance information (MII) can overcome the misalignment effects caused by the random orientations and movement of AUVs.   

Magnetic induction MIMO has been used for wireless communications and power transfer in underground and terrestrial environments. In \cite{markham2012magneto}, the tri-axis coil-based MIMO with polarization modulation is used to improve the underground communication speed to rescue trapped miners. In \cite{kisseleff2017magnetic}, the underground simultaneous wireless communications and power transfer using tri-axis coil-based MIMO is modeled by considering one information receiver and multiple energy receivers. Optimal beamforming strategies are developed to efficiently charge the energy receivers and maintain the required signal-to-noise ratio (SNR) for the information receiver. Although \cite{markham2012magneto} and \cite{kisseleff2017magnetic} use tri-axis coils, only a single information receiver is considered. Moreover, the underground wireless channel is different from the underwater environment; the underwater wireless devices can be mobile. Another research line of magnetic induction MIMO is wireless power transfer. In \cite{jadidian2014magnetic,moghadam2017node,yang2017magnetic}, the optimal beamforming strategies using magnetic induction MIMO are derived to efficiently charge wireless devices. The considered coil array is planar, which is more suitable for indoor applications rather than tiny robots due to the limited size. 

Although the tri-axis coil-based MIMO can improve reliability, it is not efficient for small AUVs due to the complexity and power consumption. In this paper, we want to answer the following two questions: 1) to guarantee the reliability, do we really need the 3$\times$3 MIMO or can we reduce the coil number by using coil selection? 2) if we use the  3$\times$3 MIMO, what is the multiplexing gain and how many users can be accommodated simultaneously? To the best of our knowledge, the above questions have not been addressed yet. In \cite{guo2017multiple,markham2012magneto,kisseleff2017magnetic} the tri-axis coils are employed, but there is no effort to reduce the coil number while maintaining the reliability. The analyses in \cite{domingo2012magnetic,gulbahar2012communication,guo2017multiple,markham2012magneto,kisseleff2017magnetic} are point-to-point communications, i.e., there is only one transmitter and one receiver. 

Different from terrestrial MIMO channels with independent fading, the channels of  tri-axis coil-based MIMO are not completely independent; they form a complete orthogonal basis in orientation to improve reliability. We analyze the UWMIC channel and emphasize these unique characteristics. Also, we discuss the coil selection strategy and the associated reliability and multiplexing gain. We find that even without using three coils simultaneously, we can achieve reliability in the high SNR regime. In this way, we can use fewer coils and the rest coils can be utilized for spatial multiplexing or energy harvesting. Finally, we explore the multiple user scenario, which is often used when there is a small swarm of AUVs with one swarm head which sends commands to other AUVs.  

The rest of this paper is organized as follows. In Section~II, the UWMIC channel is analyzed and the tri-axis coil-based communication system is introduced. In Section III, we present the reliability and multiplexing gain of different coil selection strategies and the differences of UWMIC channel and the wireless fading channel are compared. Also, the multi-user communication strategy is proposed and analyzed. In addition, we provide a channel estimation approach to update the MII in a timely manner. In Section IV, the simulation results are presented. Finally, this paper is concluded in Section V.  
\section{UWMIC Channel Characteristics}
\subsection{UWMIC Channel}
\begin{figure}[t]
	\centering
	\includegraphics[width=0.25\textwidth]{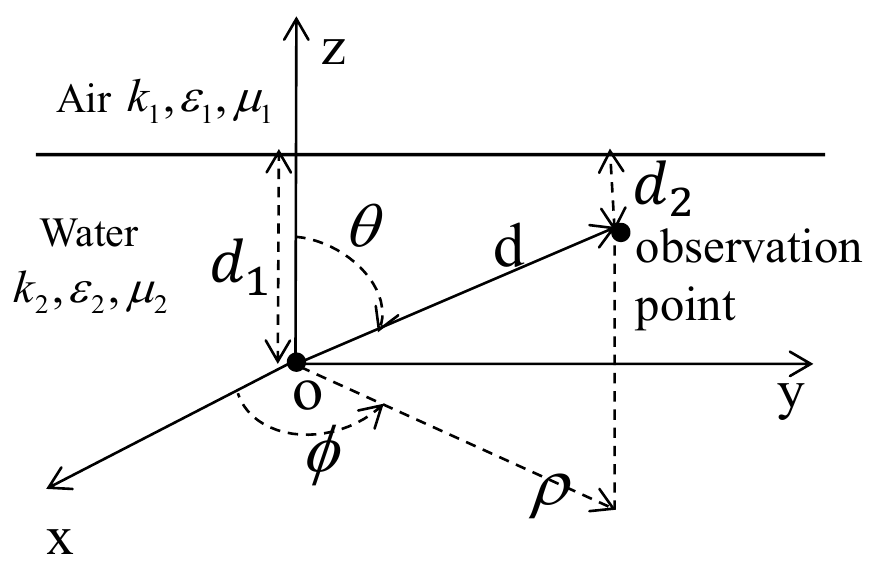}
	\vspace{-5pt}
	\caption{Cartesian and cylindrical coordinates system. The origin is the location of a transmit tri-axis coil. }
	\vspace{-10pt}
	\label{fig:sys}
\end{figure}
The UWMIC channel models are derived in \cite{domingo2012magnetic,gulbahar2012communication,guo2017multiple,guo2015channel}. The deep underwater channel models without considering the water-air interface are given in \cite{domingo2012magnetic,gulbahar2012communication}, and the shallow underwater channel models considering the lateral wave and reflections from the water-air interface are given in \cite{guo2017multiple,guo2015channel}. In the follows, we present the UWMIC model with low computation complexity than the model in \cite{guo2017multiple} and show the importance of the coil orientation.

In the following, we use both the Cartesian coordinates and the cylindrical coordinates to model the coil orientation and the magnetic field. The coordinates system is shown in Fig.~\ref{fig:sys}. For an arbitrarily orientated coil with orientation ${\bf  u}=[u_x,u_y,u_z]$, we can decompose its dipole moment into x-, y-, and z-orientated coils. 
In \cite{guo2017multiple}, an exact model is developed which consists of three layers of medium and the equations are complicated. Here, by using the method in \cite{chew1995waves}, we developed a simpler model with only the air and water media to characterize the magnetic field propagation in underwater. Since our key contribution of this paper is the coil selection and optimal transmission strategies, we defer the model to the Appendix and it is used in our numerical simulation.  Next, we use a simplified model to show how the coil orientation affects the channel quality. When the transceivers' depth are much smaller than their distance, we can obtain the magnetic fields generated by the $z-$orientated coil \cite{guo2017multiple},    	
\begin{align}
\label{equ:vertical_simp}
{\bf h}^z= [{h_{\rho}^{z}},{h_{\phi}^{z}},{h_{z}^{z}}]\approx\left[-1,0,\frac{-k_1}{k_2}\right] \frac{ji_z n_ca^2 k_1^2}{2 k_2 d^2}e^{jk_2(d_1+d_2)+jk_1d}
\end{align}
and the magnetic fields generated by the $x-$ and $y-$ orientated coils
\begin{align}
\label{equ:horizontal_simp}
&{\bf h}^{x/y}=[{h_{\rho}^{x/y}},{h_{\phi}^{x/y}},{h_{z}^{x/y}}]\approx\left[\frac {  -k_2}{k_1} \cos{\phi_{x/y}},\frac{-jk_2}{k_1^2d}\sin{\phi_{x/y}},\right.\nonumber\\
&\left.{-}\cos{\phi_{x/y}}\right]\cdot\frac{ ji_{x/y}n_c a^2 k_1^2}{2 k_2d^2}e^{jk_2d_1+jk_2d_2+jk_1d};
\end{align}
where $i$, $n_c$, and $a$ are the coil current, number of turns, and radius, respectively; $k_1=\omega\sqrt{\mu_1\epsilon_1}$ and $k_2=\omega\sqrt{\mu_2(\epsilon_2+j\sigma_2/\omega)}$ are the propagation constant of air and water, respectively; $\mu$, $\epsilon$, and $\sigma$ are permeability, permittivity, and conductivity, respectively; $d_1$, $d_2$ and $d$ are the transmitter depth, receiver depth and their distance, respectively; and $\phi_x$ and $\phi_y$ are the azimuth angles of $x-$ and $y-$ orientated coils, respectively. 

In view of \eqref{equ:vertical_simp} and \eqref{equ:horizontal_simp}, we can learn that the magnitude of the magnetic fields generated by the horizontal coil are $k_2/k_1$ times larger than those generated by the vertical coil by neglecting the azimuth angle. Since $k_2$ is around 10~dB larger than $k_1$ (the relative permittivity of water is around 81 for the considered frequency and temperature), the horizontal coils are more efficient in generating magnetic fields. 

Based on the magnetic field intensity, we can derive the mutual inductance between two arbitrarily orientated coils. Let the transmit coil orientation and the receive coil orientation be ${\bf u}_p$ and ${\bf u}_q$, which are unit vectors. Then, the mutual inductance is 
\begin{align}
\label{equ:mutual}
m_{p,q} = \mu_2 \pi a^2 n_c {\bf u}_q^t{\bf L}{\bf H}{\bf u}_p
\end{align}
where ${\bf H}=[{{\bf h}^x}^t,{{\bf h}^y}^t,{{\bf h}^z}^t]$ with unit transmit coil current and
\begin{align}
{\bf L} = \begin{bmatrix}
\cos \phi & -\sin \phi & 0           \\
\sin \phi & \cos\phi & 0    \\
0 & 0 & 1
\end{bmatrix},
\end{align}
which converts parameters in cylindrical coordinates to Cartesian coordinates. 


\subsection{Reliability and Multiplexing Gain}
The motivation of employing tri-axis coil is to increase the reliability. The single-input-single-output (SISO) with unidirectional coils (shown in the upper left in Fig.~\ref{fig:configuration}) may have very weak connection due to the orientation. 
The capacity of the SISO system can be written as
\begin{align}
\label{equ:siso_capacity}
C(m_{p,q})=\log_2\left[1+\frac{|\omega{m_{p,q}}|^2P_t}{4 r_c^2 n}\right],
\end{align}
where $r_c$ is the resistance of a coil, including coil conductive resistance and source/load resistance, $P_t$ is the overall transmit power, and $n$ is the noise power density. Here, we adopt the MI coil and channel joint model given in \cite{lin2015distributed}. As suggested by \eqref{equ:siso_capacity}, ${m_{p,q}}$ affects the capacity significantly. When ${\bf u}_q$ is perpendicular to ${\bf L}{\bf H}{\bf u}_p$, ${m_{p,q}}=0$. When ${\bf u}_q$ is parallel with ${\bf L}{\bf H}{\bf u}_p$ and ${\bf u}_p$ is horizontal with $\phi_x=0$, ${m_{p,q}}={m^{\ast}}$, which is the maximum mutual inductance that can be achieved when the positions of the transmitter and receiver are determined. Since $C(0)$ and $C(m^{\ast})$ can be drastically different, the random orientation creates significant unreliability for UWMIC. 

Note that the uncertainty of UWMIC capacity is a slow process (compared with symbol period) which is different from terrestrial fading channels. If two coil's mutual inductance is zero, it may last for several seconds and there is no connectivity.

In this paper, we define the reliability as 
\begin{definition}
	The reliability of a UWMIC system with fixed locations and random orientations is 
	\begin{align}
	{\mathcal R}=\frac{\min C(m)}{\max C(m)},~m\in[0,m^{\ast}].
	\end{align}
\end{definition}
As a result, the reliability of a SISO system with unidirectional transmit coil and receive coil is 0. 
The multiplexing gain is defined in a conventional way \cite{zheng2003diversity}.
\begin{definition}
	The multiplexing gain of a UWMIC system with fixed locations and random orientations is 
	\begin{align}
	{\mathcal M}=\lim\limits_{\frac{|\omega{m}|^2P_t}{4 r_c^2 n}\to \infty} \frac{ C(m)}{\log_2\left[\frac{|\omega{m}|^2P_t}{4 r_c^2 n}\right]}.
	\end{align} 
\end{definition}
 Thus, the multiplexing gain of SISO with unidirectional transmit coil and receive coil is 1. Next, we analyze the reliability and the multiplexing gain of different transmission and reception strategies for UWMIC.

\section{UWMIC Coil Selection Strategies}
In this section, we study different coil selection strategies as shown in Fig.~\ref{fig:configuration} to examine their reliability and multiplexing gain. For MIMO, all the three mutually perpendicular coils are utilized; there is no coil selection. For the multiple-input-single-output with coil selection (MISO-CS) and the single-input-multiple-output with coil selection (SIMO-CS), one of the transmitter and receiver uses all the three coils and the other one selects the best coil to maximize the capacity. For SISO with coil selection (SISO-CS), both the transmitter and receiver select the best coils to communicate. When the transmitter and receiver can select coils, we implicitly assume that they have the MII. The coils that are not selected can be utilized for energy harvesting, wireless sensing, among others. It can also save power without using them for communications since the active radio components are not used. We first study the single user communication and then we extend it to multiple users, i.e., one-to-many communications.  
\begin{figure}[t]
	\centering
	\includegraphics[width=0.25\textwidth]{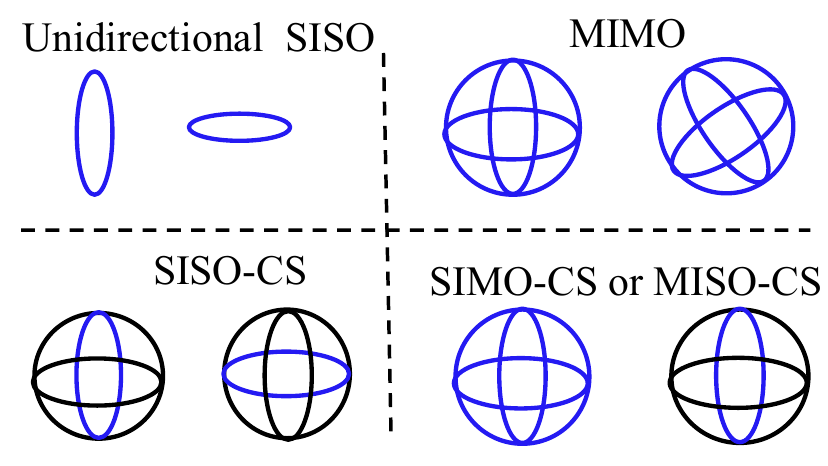}
	\vspace{-5pt}
	\caption{Coil selection configurations for UWMIC with tri-axis coils. The blue coils are selected and the black coils are not used.  }
	\vspace{-10pt}
	\label{fig:configuration}
\end{figure}
\subsection{Single-User Case}


\subsubsection{MIMO}
We assume there are one transmitter and one receiver; we denote a tri-axis coil as ${\mathcal V}_l=\{l_1,l_2,l_3\}$, which has three mutually perpendicular coils.  For each coil, let ${\bf i} = [i_x, i_y,i_z]$ and ${\bf v}= [v_x, v_y,v_z]$ denote the currents and voltages, respectively. The bandwidth of baseband signals is smaller than the joint bandwidth of the coil and wireless channel, which is equivalent to the flat fading.  By using current sources, the transmitted signal can be written as
$
{\bf i}_p= {\bf W}_p {\bf x},
$
where ${\bf W}_p$ is a 3$\times$3 matrix with $w_{pq}$ as the coefficient for the $p$th coil and the $q$th transmit symbol and ${\bf x}=[x_1,x_2,x_3]^t$ consists of the three transmit symbols.  The received signal is 
\begin{align}
{\bf v}_q=-j\omega {\bf W}_q{\bf M}_{p,q} {\bf W}_p {\bf x}+{\bf W}_q{\bf n}.
\end{align}
where $\omega$ is the angular frequency, $\bf n$ is the noise, and ${\bf M}_{p,q}$ is the mutual inductance between the transmitter and the receiver, which is 
\begin{align}
{\bf M}_{p,q} = 
\begin{bmatrix}
m_{p_1,q_1} & m_{p_1,q_2} & m_{p_1,q_3}           \\
m_{p_2,q_1} & m_{p_2,q_2} & m_{p_2,q_3}     \\
m_{p_3,q_1} & m_{p_3,q_2} & m_{p_3,q_3}  
\end{bmatrix},
\end{align} 
and ${\bf W}_q$ consists of optimal coefficients to combine the received signals. 
Also, the elements of ${\bf x}$ have the same power and it is normalized to be 1. The overall transmit power is 
$
P_t=\frac{1}{2}r_c Tr ({\bf W}_p^{\dagger}{\bf W}_p).
$
Next, we study the capacity and signal transmission strategy. 

First, we assume there is no feedback from the receiver and the transmitter equally allocate power to each of the three coils. 
When SNR is high, the multiplexing gain is fully utilized. However, when SNR is low, only the dominant channel is used. We have the following theorem
\begin{theorem}
	\label{the:mimo}
	The UWMIC channel capacity using MIMO in the high SNR regime can be approximated by 
	\begin{align}
	C_{mimo}^{high}\approx 3\log_2\left[\frac{|\omega m^{\ast}|^2P_t}{12 r_c^2n}\right]+\sum_{l=1}^{3}\log_2\lambda_l^2({\bf M}_{p,q}),
	\end{align} 
	where $\lambda_l$ is the eigenvalue of ${\bf M}_{p,q}$, and in the low SNR regime it can be approximated by
	\begin{align}
	\label{equ:mimo_low}
	C_{mimo}^{low}\approx \frac{|\omega m^{\ast}|^2P_t}{12r_c^2n}\log_2 e.
	\end{align}
\end{theorem}

\begin{proof}
	The result of $C_{mimo}^{high}$ is not surprising; the multiplexing gain is the smaller one of the transmit coil number and the receive coil number, which is three by using tri-axis coil. The proof is similar to the MIMO with terrestrial fading channels, which can be found in \cite[Chap. 2]{huang2011mimo}.  
	
	The result of $C_{mimo}^{low}$ is different from the capacity of typical fading channels. Note that, the uncertainty of UWMIC comes from the random orientation of the coils rather than multipath fading. Since magnetic induction communication does not rely on wave propagation, it barely suffers from multipath fading. 
	
	The capacity without feedback can be written as
	\begin{align}
	\label{equ:mimo_low_proof1}
	C=\log_2\left[1+\frac{\omega^2P_t}{12 r_c^2 n}tr ({\bf M}{\bf M}^{\dag})\right] {{\approx}}\frac{\omega^2P_t}{12 r_c^2 n}tr ({\bf M}{\bf M}^{\dag})\log_2 e.
	\end{align}
	The approximation is under low SNR assumption. Next, we have
	\begin{align}
	&\label{equ:mimo_low_proof2} tr ({\bf M}{\bf M}^{\dag})=\sum_{p=1}^{3}\sum_{q=1}^{3}|m_{p,q}|^2\\
	&\label{equ:mimo_low_proof3} =\sum_{p=1}^{3}\sum_{q=1}^{3}{\mu^2 \pi^2 a^4 n_c^2({\bf u}_q^t {\bf h}_p)({\bf u}_q^t {\bf h}_p)^{\dag}}\\
	&\label{equ:mimo_low_proof4} =\sum_{p=1}^{3}\mu^2 \pi^2 a^4 n_c^2\left(|{\bf h}_p{\bf u}_1|^2+|{\bf h}_p{\bf u}_2|^2+|{\bf h}_p{\bf u}_3|^2\right) \\
	&\label{equ:mimo_low_proof5}=\sum_{p=1}^{3}\mu^2 \pi^2 a^4 n_c^2 |{\bf h}_p|^2 =	\mu^2 \pi^2 a^4 n_c^2|{\bf h}^{\ast}|^2 =  |m^{\ast}|^2,
	\end{align}
	where ${\bf h}_p={\bf L}{\bf H}{\bf u}_p$ and ${\bf h}^{\ast}$ is the largest magnetic field that can be generated by the optimal orientated transmit coil with unit current, i.e., there is no orientation loss. In \eqref{equ:mimo_low_proof3} we use the definition of the mutual inductance. In \eqref{equ:mimo_low_proof4}, we consider ${\bf u}_1$, ${\bf u}_2$, and ${\bf u}_3$ form an orthogonal coordinates system and ${\bf h}_p$ is a vector in the system. By projecting ${\bf h}_p$ onto each axis and norm the coordinates, we have the length of ${\bf h}_p$. Also, we can find 
	\begin{align}
	\sum_{p=1}^{3}|{\bf h}_p|^2=|{\bf L}{\bf H}|^2\left({\bf u}_1^2+{\bf u}_2^2+{\bf u}_3^2\right) =|{\bf L}{\bf H}|^2.
	\end{align} 
	Thus, there is no orientation loss and the mutual inductance is equivalent to the optimally orientated single unidirectional coil system.
	By substituting \eqref{equ:mimo_low_proof5} into \eqref{equ:mimo_low_proof1}, we can prove \eqref{equ:mimo_low}.
\end{proof}

From Theorem \ref{the:mimo} we can see the UWMIC using MIMO can increase the diversity to three in high SNR regime, but it has the same order of capacity as the best SISO in low SNR regime, which means we can only increase reliability but cannot increase capacity using MIMO in low SNR regime. Also, by comparing the low SNR approximation of \eqref{equ:siso_capacity} and \eqref{equ:mimo_low}, we find that there is a 5 dB (3 times) difference. Without MII feedback, the MIMO capacity is 5 dB lower than the optimal SISO (with mutual inductance ${m}^{\ast}$) in the low SNR regime. 

When there is feedback of MII, we can use SVD decomposition ${\bf M}_{p,q}={\bf U}_{p,q}{\bf \Lambda}_{p,q} {\bf V}_{p,q}^{\ast}$ and let ${\bf W}_q={\bf U}_{p,q}^{\ast}$ and ${\bf W}_p={\bf V}_{p,q}$. Then, we can obtain three independent channels and use water-filling algorithm to allocate power based on ${\bf \Lambda}_{p,q}$. In this way, in the high SNR regime, we can obtain the multiplexing gain of three and in the low SNR regime the MIMO and the optimal SISO have the same performance since we can always optimally allocate power. 

In general, with or without MII the reliability of MIMO using tri-axis coil is always 1 since the capacity is independent of coil orientation and the maximum and the minimum capacity are the same. For high SNR, the multiplexing gain is three for both with and without MII. When SNR is low, the MIMO without MII lost 5 dB in capacity compared with the optimal SISO. 

Although the reliability of MIMO is 1, there are two drawbacks. First, both the transmitter and the receiver have to employ three RF chains, which is power-hungry. Also, the coils are perpendicular to each other in orientation and with high probability only one or two unidirectional coil(s) can receive a large portion of the received power and, thus, it is not efficient to turn on all the RF chains.  Second, AUVs are moving in swarms, it is desirable to communicate with the swarm members simultaneously to coordinate their formation. Using all the coils to communicate with a single AUV is not efficient. Next, we study coil selection strategies to show that we can also achieve reliability using fewer coils. 
\subsubsection{SISO-CS}
In the SISO-CS scenario, we assume the transmitter and the receiver use tri-axis coils, but they select the best unidirectional coil to send and receive signals. Both the transmitter and receiver have the MII. To maximize the channel capacity, the optimal transmit coil $p^{\ast}$ and the optimal receive coil $q^{\ast}$ are
\begin{align}
\{p^{\ast},q^{\ast}\}=\underset{p\in {\mathcal V}_1,q \in{\mathcal V}_2 }{\mathrm{argmax}} |m_{p,q}|.
\end{align}
The reliability and data rate depend on $|m_{p^{\ast},q^{\ast}}|$. Different from MIMO, the transceiver of SISO-CS only uses one coil, which cannot form a complete orthogonal coordinates system, and thus it becomes orientation-dependent, i.e., $|m_{p^{\ast},q^{\ast}}|$ is a random number. Here, we can analytically find the upper bound and lower bound of the capacity, which are given in the following theory:
\begin{theorem} 
	\label{the:siso}
  When MII is available at transmitter and receiver, the capacity of SISO-CS in the high SNR and low SNR regime are 
	\begin{align}
	\log\left[\frac{|\omega m^{\ast}|^2P_t}{36r_c^2n} \right]\leq &C_{siso}^{high}\leq \log\left[\frac{|\omega m^{\ast}|^2P_t}{4r_c^2n}\right];\\
	\frac{|\omega m^{\ast}|^2P_t}{36 r_c^2n}\log_2 e  \leq &C_{siso}^{low}\leq \frac{|\omega m^{\ast}|^2P_t}{4r_c^2n}\log_2 e.
	\end{align}
\end{theorem}

\begin{proof}
	According to \eqref{equ:horizontal_simp}, when the positions of the transmitter and receiver are fixed there is an optimal orientation ${\bf u}_t^{\ast}$ of the transmit coil to generate the maximum magnetic field at the receiver, which is the same direction as the transmit coil to create $m^{\ast}$. Also, the magnetic field at the receiver has a direction and there is an optimal orientation ${\bf u}_r^{\ast}$ to receive the maximum amount of power. Consider the transmit tri-axis coil orientation is ${\bf U}_t=[{\bf u}_{t1},{\bf u}_{t2},{\bf u}_{t3}]$ and the receive tri-axis coil orientation is ${\bf U}_r=[{\bf u}_{r1},{\bf u}_{r2},{\bf u}_{r3}]$. ${\bf U}_t$ and ${\bf U}_r$ are orthogonal matrices with unit row and column vectors and the row/column vectors can form two orthogonal coordinates system basis. ${\bf u}_t^{\ast}$ and ${\bf u}_r^{\ast}$ can be considered as two vectors in the two orthogonal coordinates systems, respectively. Without loss of generality, we assume ${\bf u}_t^{\ast}=[1,0,0]^t$, then we have to select the coil with the largest x-direction component, i.e., the coil with the largest element in the first row of  ${\bf U}_t$. 
	
	The best case is when the largest element is 1; it cannot be larger than it because both the row vector and the column vector are unit. The worst case is the three elements in the first row are equal and the absolute value of the largest element is $\sqrt{3}/3$. It cannot be smaller than $\sqrt{3}/3$ since if such a number exists, one of the other two elements have to be larger than $\sqrt{3}/3$ to maintain the unit vector property and the coil with that element is selected. The same analysis can be applied to the receiver. 
	
	As a result, the upper bound is achieved when $|m_{p^{\ast},q^{\ast}}|=|m^{\ast}|$. The lower bound is achieved when both the transmit coil and receive coil orientations scaled by $\sqrt{3}/3$, which results in  $|m_{p^{\ast},q^{\ast}}|=|m^{\ast}|/3$. 
\end{proof}


Generally, for the SISO-CS with MII, the reliability is  1 and 1/9 in the high SNR regime and low SNR regime, respectively. The multiplexing gain is always 1. Without MII, SISO-CS cannot be performed since the transmitter cannot select the best coil.  
\subsubsection{SIMO-CS and MISO-CS}
The SIMO-CS is employed when the transmitter selects the best  transmit coil, while the receiver uses all the coils. The MISO-CS is used when the transmitter uses all the coils and the receiver selects the best coil. The capacity can be bounded by the following theory:
\begin{theorem}
	\label{the:simo}
	When MII is available at transmitter and receiver, the capacity of SIMO-CS and MISO-CS share the same upper bound and lower bound, which are
	\begin{align}
	\log_2\left[\frac{|\omega m^{\ast}|^2P_t}{12r_c^2n} \right]\leq &C_{m,s}^{high}\leq \log_2\left[\frac{|\omega m^{\ast}|^2P_t}{4r_c^2n}\right];\\
	\frac{|\omega m^{\ast}|^2P_t}{12 r_c^2n}\log_2 e  \leq &C_{m,s}^{low}\leq \frac{|\omega m^{\ast}|^2P_t}{4r_c^2n}\log_2 e.
	\end{align}
\end{theorem}
\begin{proof}
	We provide the proof for SIMO-CS and the capacity of MISO-CS can be proved in a similar way. First, the capacity for SIMO-CS is 
	\begin{align}
	\label{equ:simo_capacity}
	C=\log_2\left(1+\frac{\omega^2P_t\sum_{q=1}^{3} |m_{p^{\ast},q}|^2}{4r_c^2 n}\right)
	\end{align}
	where 
	\begin{align}
	\sum_{q=1}^{3}|m_{p^{\ast},q}|^2&=\sum_{q=1}^{3}{\mu^2 \pi^2 a^4 n_c^2({\bf u}_q^t {\bf h}_{p^{\ast}})({\bf u}_q^t {\bf h}_{p^{\ast}})^{\dag}}\\
	& =\mu^2 \pi^2 a^4 n_c^2\left(|{\bf h}_{p^{\ast}}{\bf u}_1|^2+|{\bf h}_{p^{\ast}}{\bf u}_2|^2+|{\bf h}_{p^{\ast}}{\bf u}_3|^2\right) \\
	& =\mu^2 \pi^2 a^4 n_c^2|{\bf h}_{p^{\ast}}|^2.
	\end{align}
	In preceding discussions, we have shown that the minimum $|{\bf h}_{p^{\ast}}|$ is $\sqrt{3}/3$ times of the maximum $|{\bf h}_{p^{\ast}}|$. Thus, $|{m^{\ast}}|^2/3\leq \sum_{q=1}^{3}|m_{p^{\ast},q}|^2\leq |{m^{\ast}}|^2$. By substituting the upper bound and lower bound of $\sum_{q=1}^{3}|m_{p^{\ast},q}|^2$ into \eqref{equ:simo_capacity} and using the high SNR and low SNR approximations, we can prove the theory. 
\end{proof}  

If the transmitter does not have MII, it cannot select the best coil and the SIMO-CS becomes MIMO without MII, since the transmitter equally allocates its power to the three coils. For MISO-CS without MII, both the $m^{\ast}$ and transmit power for each coil need to be scaled by $1/3$ for the capacity lower bound of MISO-CS, while the upper bound does not change compared with the scenario with MII. Also, we notice that the capacity bounds of MISO-CS without MII is the same as the SISO-CS with MII, which are given in Theory \ref{the:siso}. 

As a result, for SIMO-CS and MISO-CS with MII, the reliability is 1 and 1/3 in the high SNR regime and low SNR regime, respectively. For SIMO-CS without MII, the reliability is 1 both in the high SNR regime and low SNR regime. For MISO-CS without MII, the reliability is 1 and 1/27 in the high SNR regime and low SNR regime, respectively. The multiplexing gain is always 1; there is no power gain which is different from typical fading channels. 
\subsection {Multiuser Case}
For real-time AUVs communications, low-latency data transmission is as important as the overall throughput. Therefore, here our objective is to accommodate more users. Consider that there is a swarm of AUVs with a swarm head, which controls the formation and tasks of the other AUVs. Since each of them only has one tri-axis coil, the multiplexing gain cannot be larger than three. As discussed in preceding part, the SISO-CS, SIMO-CS, and MISO-CS solutions can also provide reliable signal. Next, we study the coil selection strategies to accommodate two and three users .

When there are four AUVs in a swarm, the swarm head communicates with the other three simultaneously. Given the mutual inductance ${\bf M}_i =[{\bf m}_{i,1}^t, {\bf m}_{i,2}^t,{\bf m}_{i,3}^t]^t$ of the $i$th receiver and the transmitter, where ${\bf m}_{i,q}$ is the mutual inductance between the $q$th receive coil and the transmit coil. 

First, each receiver selects an optimal receiving coil based on the MII, 
\begin{align}
\{q^{\ast}\}=\underset{q=1,2,3}{\arg\max} ||{\bf m}_{p,q}||.
\end{align}
Then, we construct three unit vectors ${\bf u}_p$ (p=1,2,3), which are orthogonal to ${\bf m}_{p,q^{\ast}}$ (p=1,2,3). For example, to find ${\bf u}_1$, we construct a new matrix $[{\bf m}_{2,q^{\ast}}^t,{\bf m}_{3,q^{\ast}}^t]^t={\bf U}_1{\bf \Lambda}_1{\bf V}_1^{\dag}$, where the right-hand-side is the singular value decomposition of the matrix. Here ${\bf V}_1$ is a 3$\times$3 matrix and the last column is orthogonal to ${\bf m}_{2,q^{\ast}}$ and ${\bf m}_{3,q^{\ast}}$, which is ${\bf u}_1^t$. The transmit signal is ${\bf x}=x_1{\bf u_1}+x_2{\bf u}_2+x_3{\bf u}_3$, where $x_q$ is the information sending to the $q$th receiver. 
Considering the UWMIC channel, the received signal at the $q$th receiver is
\begin{align}
v_q = -j\omega {\bf m}_{p, q^{\ast}}{\bf x}+n=-j\omega x_q{\bf m}_{p, q^{\ast}}{\bf u}_q+n. 
\end{align} 
In this way, we have three separated data channel for each receiver. The transmission power is $P_t=r_c{\bf x}^{\dag}{\bf x}/2$. The lower bound and upper bound of the per-user capacity is the same as the MISO-CS without mutual inductance information, since the power is equally allocated to each symbol and each receiver only uses one coil.   

When there are two receivers ${\mathcal V}_{1}$ and ${\mathcal V}_{2}$ in the group, without loss of generality, we assume two channels are allocated to ${\mathcal V}_{1}$ and the other one channel is allocated to ${\mathcal V}_{2}$ to fully utilize the diversity. Then, ${\mathcal V}_{1}$ selects the receive coils  $\{q_1^{\ast},q_2^{\ast}\}={argmax}||{\bf m}_{p,q}||~ {q=1,2,3}$ and ${\mathcal V}_{2}$ selects the receive coil with the largest one $||{\bf m}_{p,q}||, q=1,2,3$. The method used for three receivers is still workable. We can construct ${\bf u}_i$ using the three mutual inductance vectors and send $x_1$ and $x_2$ to ${\mathcal V}_{1}$ and $x_3$ to ${\mathcal V}_{2}$. The received signal in ${\mathcal V}_{2}$ is the same as the three receivers case, while the received signal in ${\mathcal V}_{1}$ is
\begin{align}
v_{q1} = -j\omega  \begin{bmatrix}
{\bf m}_{p, q_1^{\ast}}\\
{\bf m}_{p, q_2^{\ast}}
\end{bmatrix}
{\bf x}+n=\begin{bmatrix}
-j\omega x_1{\bf m}_{p, q_1^{\ast}}{\bf u}_1+n\\
-j\omega x_2{\bf m}_{p, q_2^{\ast}}{\bf u}_2+n
\end{bmatrix}
\end{align}
The capacity of ${\mathcal V}_{2}$ is the same as the SISO-CS with mutual inductance information, while the capacity of ${\mathcal V}_{1}$ is the same as the 2$\times$2 MIMO without MII. 

For multiuser scenario with the MII, the reliability depends on the coil selection strategy and we can always achieve 1 at high SNR. The multiplexing gain is three since we use three data streams. In this way, we can maintain the reliability and serve more users. 
%
\subsection{Mutual Inductance Estimation}
Mutual inductance is an indicator of the coupling strength between two coils. In preceding discussions, we rely on the accurate estimate of mutual inductance to optimally select coils. Next, we introduce an approach to efficiently and timely estimate the mutual inductance for AUVs.

Consider there are one transmitter with coils 1 to 3 and $c_n \in \{1,2,3\}$ receivers with coils 4 to $3(c_n+1)$. According to Kirchhoff's Voltage Law, without loss of generality, the impedance, currents, and voltages in coil $q$ can be written as
\begin{align}
\label{equ:impedance}
r_c i_q +\sum_{p=1, p\neq q}^{3(c_n+1)} j\omega m_{p,q} i_p +n= v_q.
\end{align}
We assume all the coils have the same parameters, i.e., the impedance $r_c$. There are $3(c_n+1)$ coils and we have the above equation for each one. There are $(3c_n+2)(3c_n+3)$ $m_{p,q}$; since $m_{p,q}=m_{q,p}$, only $(3c_n+2)(3c_n+3)/2$ distinct mutual inductances. Out of the $(3c_n+2)(3c_n+3)/2$ $m_{p,q}$, each transceiver has the knowledge of the mutual inductance among its own three coils, which are 0 since the coils are perpendicular in orientation. Also, the mutual inductance among receivers can be neglected due to the spatial diversity which can be guaranteed by the swarm formation. Finally, the number of mutual inductance we need to estimate reduces to $(3c_n+2)(3c_n+3)/2-3(c_n+1)-9\binom{c_n}{2}=9c_n$. 

Next, we focus on the coils of the receivers and rewrite \eqref{equ:impedance} in matrix form 
\begin{align}
\label{equ:estimation1}
{\begin{bmatrix}
	{\bf M}_{p,1}&\cdots&0\\
	\vdots&\vdots &\vdots\\
	0&\cdots&{\bf M}_{p,c_n}
	\end{bmatrix}}
\begin{bmatrix}
{\bf i}_p\\
\vdots\\
{\bf i}_p
\end{bmatrix}=\frac{1}{j\omega}
\begin{bmatrix}
-{\bf Z} {\bf i}_1\\
\vdots\\
-{\bf Z} {\bf i}_{c_n}
\end{bmatrix}
\end{align}
where ${\bf Z}$ is a diagonal matrix with $r_c$ as the diagonal elements. From the above equation, we have $9c_n$ unknown variables but only $3c_n$ linear equations, which cannot solve the problem uniquely. To address the challenge, we consider there are three time slots and in each time slot the transmitter uses different ${\bf i}_p$. In this way, the mutual inductance in \eqref{equ:estimation1} can be solved
\begin{align}
\label{equ:estimation2}
{\bf M}_{p,l}=\frac{j}{\omega}
\begingroup 
\setlength\arraycolsep{1pt}
\begin{bmatrix}
{\bf Z} {\bf i}_l(t_1)&{\bf Z} {\bf i}_l(t_2)&{\bf Z}{\bf i}_l(t_3)\\
\end{bmatrix}
\begin{bmatrix}
{\bf i}_p(t_1)&{\bf i}_p(t_2)&{\bf i}_p(t_3)\\
\end{bmatrix}^{-1}
\endgroup
\end{align}
where $l=1,2,3$ and ${\bf i}(t_l)$ is the measured current in $l$th time slot. To ensure that the inverse of the current matrix exists, the transmitter has to transmit orthogonal currents vectors, which can be constructed by the Gram–Schmidt process. All the transmitter and receivers have the knowledge of ${\bf i}_1$ at $t_1$ to $t_3$. The receiver knows its own currents at $t_1$ to $t_3$, then it can estimate ${\bf M}_{1,l}$ locally. Next, instead of sending its currents back to the transmitter, the $l$th receiver sends the estimated ${\bf M}_{1,l}$. Then, both the transmitter and receiver make their optimal decisions to transmit and receive signals.

\section{Performance Evaluation}
In this section, we numerically simulate the proposed coil selection strategies and compare the performances. First, for the signal user case, we consider the transmitter is located at the origin $[0,0,0]$ and the receiver is located at $[0,5,0.2]$(m); the coordinate system is shown in Fig.~\ref{fig:sys}. The depth of the transmitter and receiver are 0.5m and 0.3m, respectively. The relative permeability and permittivity of water are 1 and 81, respectively. The relative permeability and permittivity of air are 1 and 1, respectively. The lake/river water conductivity is 0.1S/m. The frequency is 1MHz. The coil radius, number of turns, and impedance are 0.05m, 10, and 0.5$\Omega$, respectively. The background noise density is -140 dBm/Hz \cite{shaw2006experimental}. By using the magnetic field formulas in the Appendix, we can find ${\bf H}$. Then, we randomly generate the orientation of unidirectional and tri-axis coils using the method in \cite{guo2017multiple}. 
\begin{figure}[t]
	\centering
	\includegraphics[width=0.31\textwidth]{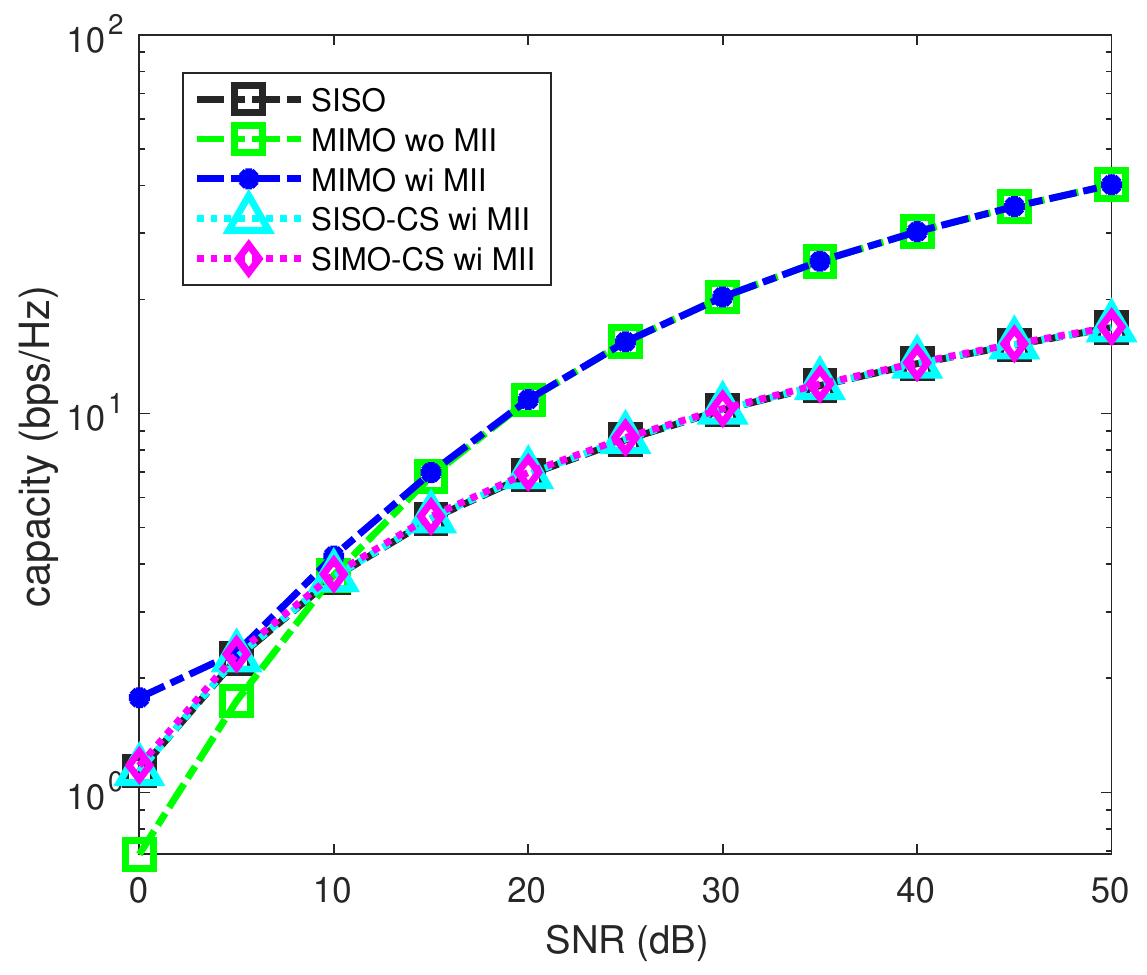}
	\vspace{-5pt}
	\caption{Simulated upper bound of capacity. }
	\vspace{-10pt}
	\label{fig:max}
\end{figure}
\begin{figure}[t]
	\centering
	\includegraphics[width=0.31\textwidth]{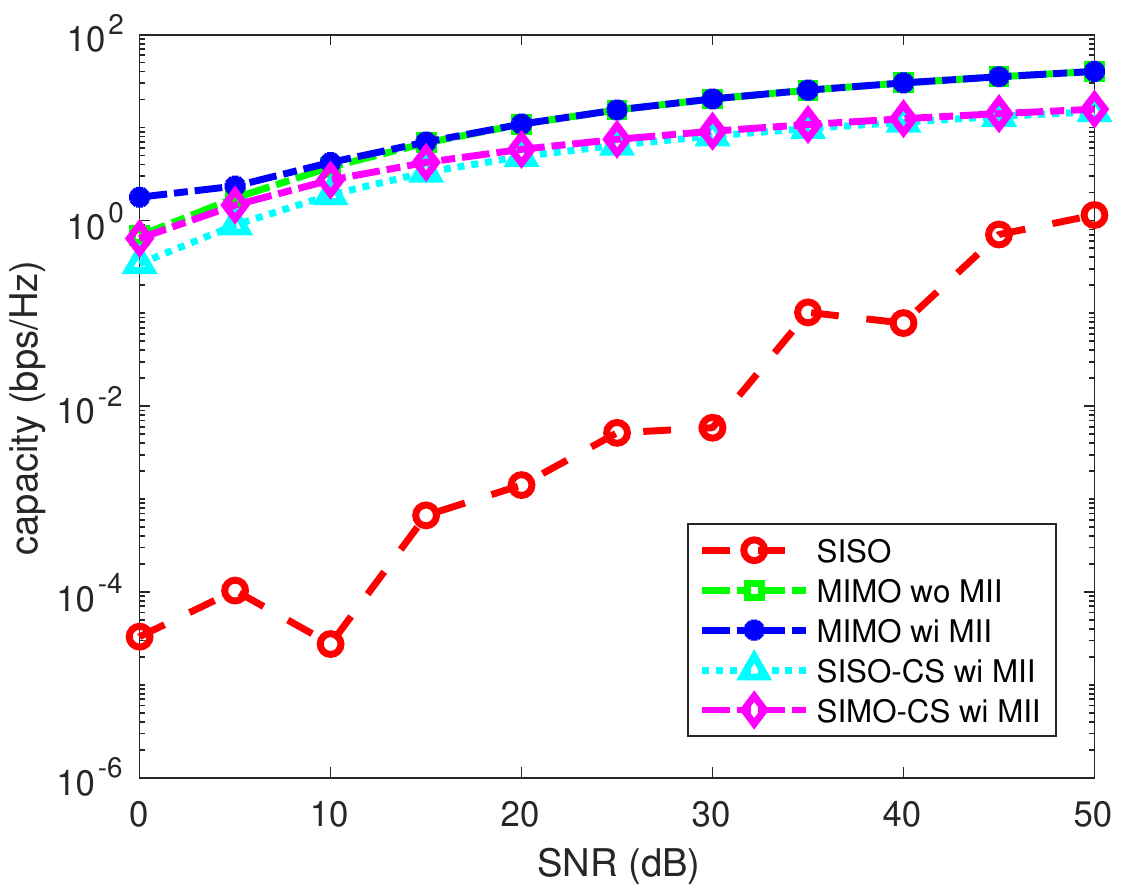}
	\vspace{-5pt}
	\caption{Simulated lower bound of capacity. }
	\vspace{-10pt}
	\label{fig:min}
\end{figure}

In Fig.~\ref{fig:max} and Fig.~\ref{fig:min}, we show the simulated upper bound and lower bound of the capacity for SISO, MIMO with MII, MIMO without MII, SISO-CS with MII, and SIMO-CS with MII. For the upper bound, we notice in the low SNR regime, all the configurations have similar capacity, which happens when the transmit coils and receive coils are well aligned. Also, the MIMO without MII is a little smaller than other configurations, because the power is divided to three coils. In the high SNR regime, the MIMO systems have higher diversity and the capacity is around three times larger than other configurations. In addition, the SIMO-CS with MII does not display any power gain, which is different from typical fading channels. For the lower bound, in low SNR regime, the capacity of MIMO with MII is about 9 times larger than the SISO-CS with MII. Also, the SIMO-CS with MII has the same performance as the MIMO without MII, which has been predicted in Theory \ref{the:mimo} and Theory \ref{the:simo}. 
\begin{figure}[t]
	\centering
	\includegraphics[width=0.31\textwidth]{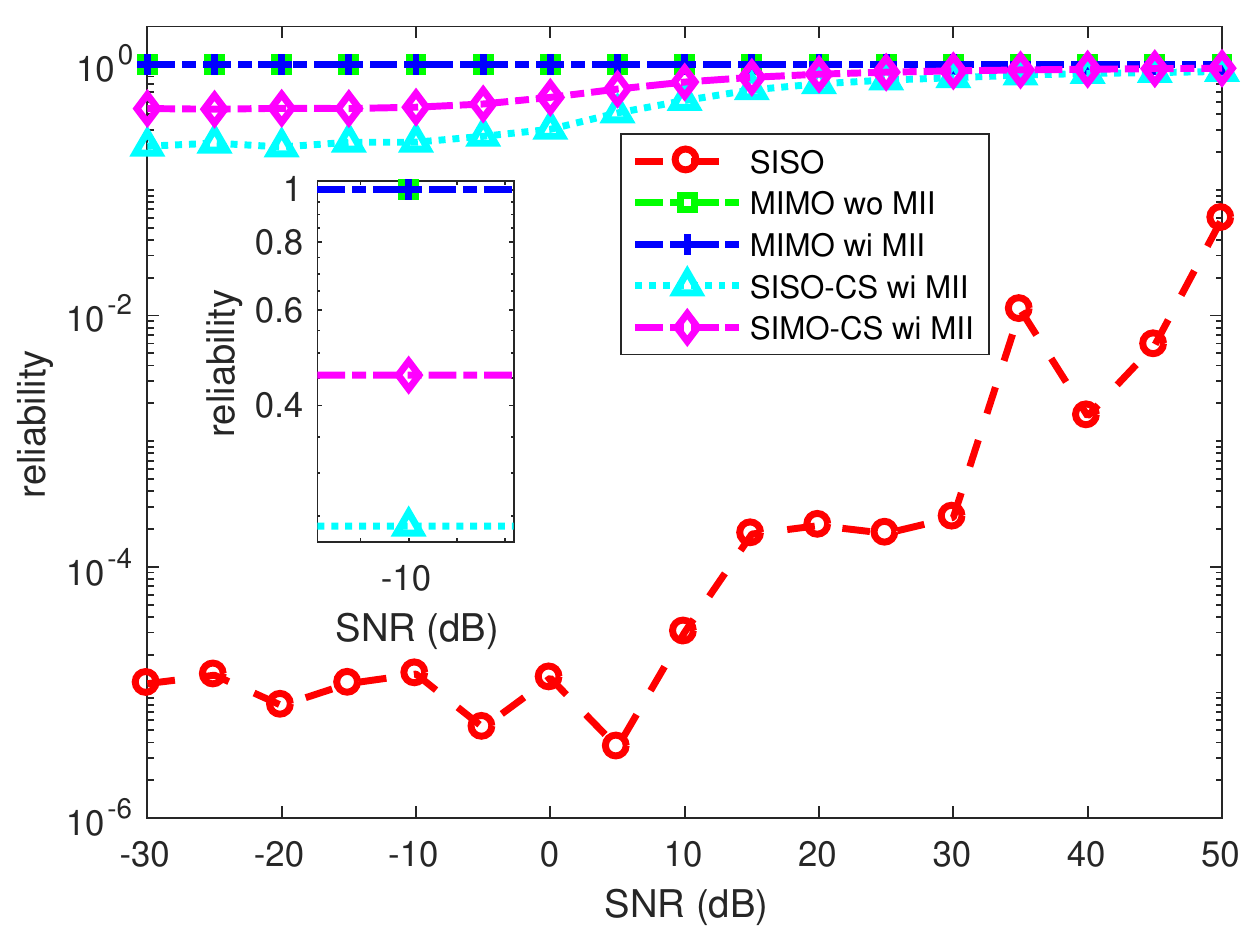}
	\vspace{-5pt}
	\caption{Effect of SNR on the reliability. }
	\vspace{-10pt}
	\label{fig:reliability}
\end{figure}

The reliability is shown in Fig.~\ref{fig:reliability}. The SISO is highly unreliable. The MIMO is always reliable no matter with or without MII. The SISO-CS and SIMO-CS have low reliability at low SNR, but in the high SNR regime, their reliability become 1. In Fig.~\ref{fig:multi}, the reliability of a user in a four AUV swarm is shown. One AUV is located at the origin with depth 0.5m, the other three are located at the x-, -x-, and y-axis with distance 5m and depth 0.3m. As we can see, the reliability is small at low SNR, but converges to 1 as SNR increases, which means in the high SNR regime, we can we can obtain multiplexing gain as well as reliability. 
  \begin{figure}[t]
  	\centering
  	\includegraphics[width=0.31\textwidth]{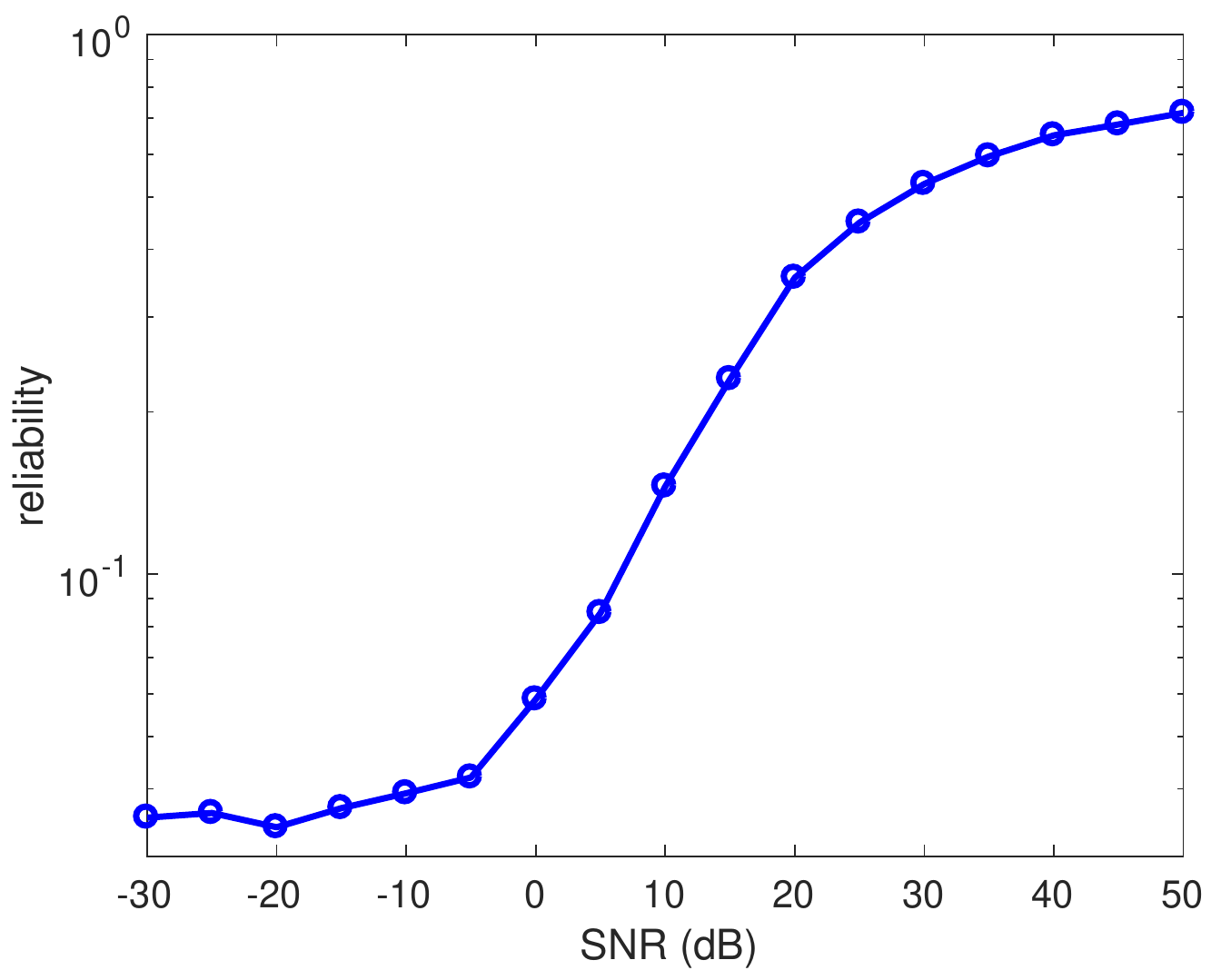}
  	\vspace{-5pt}
  	\caption{ Reliability of three receivers. There are three receivers distributed on the x-, -x-, and y-axis with the same distance from the transmitter.}
  	\vspace{-10pt}
  	\label{fig:multi}
  \end{figure}
In Fig.~\ref{fig:estimation}, the effect of MII estimation error is shown. We compare the performance with perfect MII and estimated MII. When the SNR is low, the noises have strong effects and the reliability of all the configurations are significantly reduced. As SNR increases, the estimation error becomes small, and the reliability converges to 1. 
  \begin{figure}[t]
	\centering
	\includegraphics[width=0.31\textwidth]{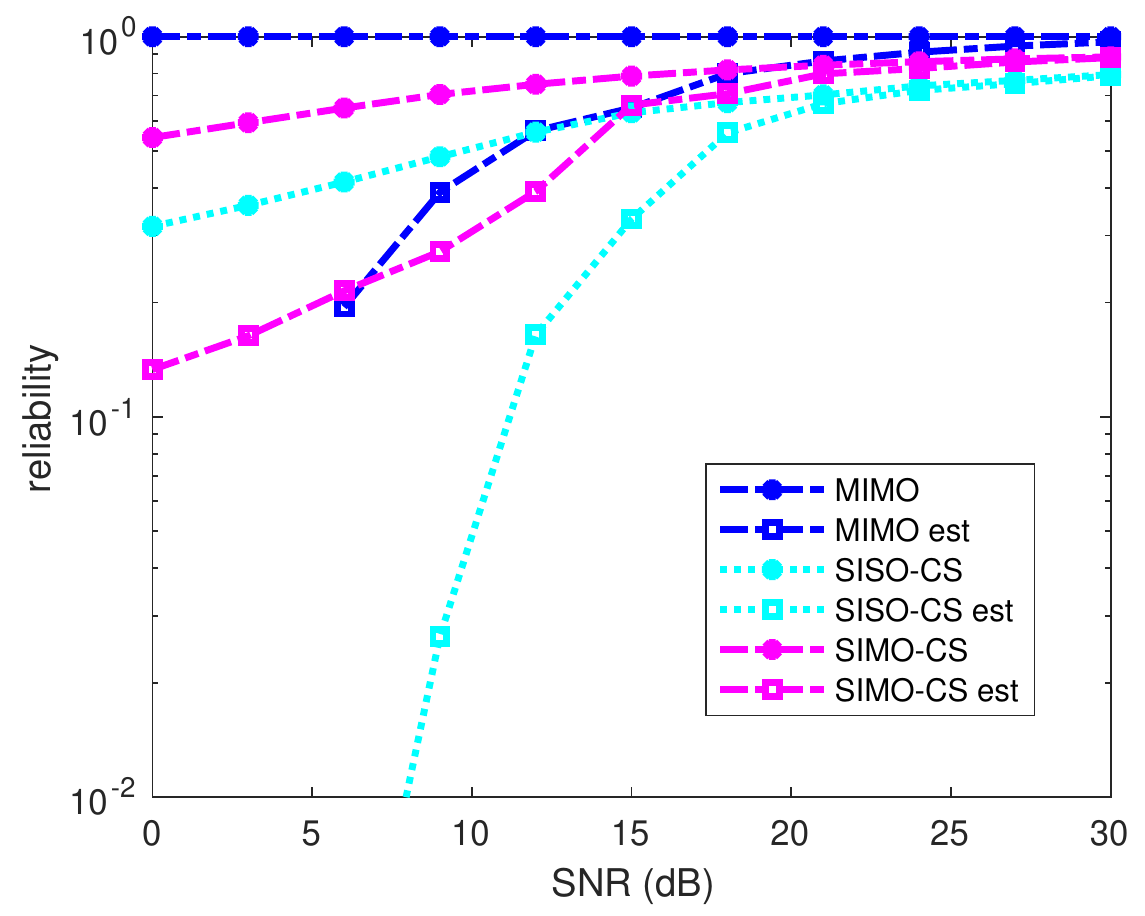}
	\vspace{-5pt}
	\caption{Effect of the estimated MII error. The est stands for estimated. Low SNR is associated with large error, while high SNR is associated with small error.}
	\vspace{-10pt}
	\label{fig:estimation}
\end{figure}
\section{Conclusion}
Underwater magnetic induction communications (UWMIC) with tri-axis coils can provide reliable connections. In this paper, we reduce the number of coils by leveraging the orientation diversity. Our results show that in the high SNR regime, we can simply use SISO-CS, i.e., one transmit coil and one receive coil, to achieve reliable communications. Also, we can increase the user number by using coil selection, which is important for real-time communications for underwater robotic networks. 
\section*{Appendix}
In Section II-A, we derive the magnetic field generated by the coils in underwater. In \cite[Chap. 2.3]{chew1995waves}, the $h_z^z$ and $h_z^{x/y}$ are given explicitly. The transverse components can be derived using $h_z^z$ and $h_z^{x/y}$ based on Maxwell equations. Then, we have the magnetic field in the underwater environment,
\allowdisplaybreaks
\begin{align}
&{h}_{\rho}^z=\zeta_1 \int_{-\infty}^{\infty}d k_{\rho}\frac{k_{\rho}^2}{k_{2z}}H_1^{(1)}(k_{\rho}\rho)[jk_{2z}\zeta_2-jk_{2z}\zeta_3],\\
&h_{z}^z = -\zeta_1 \int_{-\infty}^{\infty}d k_{\rho}\frac{k_{\rho}^3}{k_{2z}}H_0^{(1)}(k_{\rho}\rho)[\zeta_2+\zeta_3],h_{\phi}^z=0,\\
&h_{\rho}^{x/y}=j\zeta_1\cos\phi_{x/y}\int_{-\infty}^{\infty}dk_{\rho}{H_1^{(1)}}'(k_{\rho}\rho)[jk_{2z}\zeta_2-jk_{2z}\zeta_3] \\
&-\zeta_4 \cos\phi_{x/y}\int_{-\infty}^{\infty}dk_{\rho}H_1^{(1)}(k_{\rho}\rho)\zeta_5/({k_{2z}}\rho),\\
&h_{\phi}^{x/y}=-j\zeta_1\sin\phi_{x/y}\int_{-\infty}^{\infty}dk_{\rho}{H_1^{(1)}}(k_{\rho}\rho)[jk_{2z}\zeta_2-jk_{2z}\zeta_3]/\rho\\
&+\zeta_4\sin\phi_{x/y}\int_{-\infty}^{\infty}dk_{\rho}{H_1^{(1)}}'(k_{\rho}\rho)\zeta_5/k_{2z},\\
&h_{z}^{x/y}=j\zeta_1\cos\phi_{x/y}\int_{-\infty}^{\infty}dk_{\rho}k_{\rho}^2H_1^{(1)}(k_{\rho}\rho)[\zeta_2+\zeta_3],\\
&\zeta_1=\frac{ji_z a^2 n_c }{8},~~\zeta_2=e^{jk_{2z} z} ,\zeta_3=R_{21}^{TE} e^{-jk_{2z} z+2jk_{2z}d_1},\\
&\zeta_4=\frac{j\omega^2\epsilon_2\mu_2 i_{x/y}a^2 n_c}{8}, \zeta_5= \zeta_2+R_{21}^{TM}e^{-jk_{2z} z+2jk_{2z}d_1},\\
&{H_1^{(1)}}'(k_{\rho}\rho)=k_{\rho}\left[\frac{{H_1^{(1)}}(k_{\rho}\rho)}{k_{\rho}\rho}-H_2^{(1)}(k_{\rho}\rho)\right],\\
&R_{21}^{TE}=\frac{\mu_1 k_{2z}-\mu_2 k_{1z} }{\mu_1 k_{2z}+\mu_2 k_{1z}},~R_{21}^{TM}=\frac{\epsilon_1 k_{2z}-\epsilon_2 k_{1z} }{\epsilon_1 k_{2z}+\epsilon_2 k_{1z}}\\
&k_{2z}=\sqrt{k_2^2-k_{\rho}^2}, |\phi_x-\phi_y|=\pi/2, z=d_1-d_2
\end{align} 
where $H_n^{(1)}(x)$ is the Hankel function of the first kind with order $n$. 
\bibliographystyle{IEEEtran}
\bibliography{guo}
\end{document}